\documentclass[runningheads]{llncs}
\usepackage{graphicx}
\usepackage{amsfonts, amssymb, amsmath}
\usepackage[pdfencoding=auto]{hyperref}
\newcommand{\ab}{$(\alpha,\beta)$}
\newcommand{\abd}{$(\alpha,\beta)$-Dynamics }

\newtheorem{prop}{Property}

\begin{document}
\title{Threshold-based Network Structural Dynamics\thanks{This is an extended version of a post-print submitted at SIROCCO 2021, containing all proofs. The final authenticated version is available online at https://doi.org/10.1007/978-3-030-79527-6\_8. Evangelos Kipouridis received funding from the European Union's Horizon 2020 research and innovation program under the Marie Skłodowska-Curie grant agreement No 801199. \includegraphics[height=0.03\textwidth, width=0.07\textwidth]{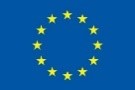} Evangelos Kipouridis is also supported by Thorup's Investigator Grant 16582, Basic Algorithms Research Copenhagen (BARC), from the VILLUM Foundation. \includegraphics[height=0.04\textwidth, width=0.07\textwidth]{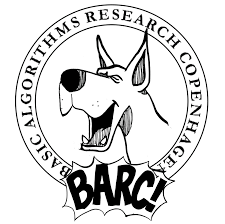}}}
%
%
\author{Evangelos Kipouridis\inst{1}\orcidID{0000-0002-5830-5830} \and
Paul G. Spirakis\inst{2,3}\orcidID{0000-0001-5396-3749} \and
Kostas Tsichlas\inst{3}\orcidID{0000-0003-4107-9520}}
\authorrunning{E. Kipouridis et al.}
%
\institute{\scriptsize{Basic Algorithms Research Copenhagen (BARC), University of Copenhagen, Copenhagen, Denmark, Universitetsparken 1, 2100\\ \email{kipouridis@di.ku.dk} \and
Department of Computer Science, University of Liverpool, Ashton Str., Liverpool L69 3BX, UK\\ \email{P.Spirakis@liverpool.ac.uk} \and
Computer Engineering and Informatics Department, University of Patras, Patras, Greece
\email{\{ktsichlas,spirakis\}@ceid.upatras.gr}}}
\maketitle              
\begin{abstract}
The interest in dynamic processes on networks is steadily rising in recent years. In this paper, we consider the $(\alpha,\beta)$-Thresholded Network Dynamics ($(\alpha,\beta)$-Dynamics), where $\alpha\leq \beta$, in which only structural dynamics (dynamics of the network) are allowed, guided by local thresholding rules executed by each node. In particular, in each discrete round $t$, each pair of nodes $u$ and $v$ that are allowed to communicate by the scheduler, computes a value $\mathcal{E}(u,v)$ (the potential of the pair) as a function of the local structure of the network at round $t$ around the two nodes. If $\mathcal{E}(u,v) < \alpha$ then the link (if it exists) between $u$ and $v$ is removed; if $\alpha \leq \mathcal{E}(u,v) < \beta$ then an existing link among $u$ and $v$ is maintained; if $\beta \leq \mathcal{E}(u,v)$ then
a link between $u$ and $v$ is established if not already present.

The microscopic structure of $(\alpha,\beta)$-Dynamics appears to be simple, so that we are able to rigorously argue about it, but still flexible, so that we are able to design meaningful microscopic local rules that give rise to interesting macroscopic behaviors. Our goals are the following: a) to investigate the properties of the $(\alpha,\beta)$-Thresholded Network Dynamics and b) to show that $(\alpha,\beta)$-Dynamics is expressive enough to solve complex problems on networks. 

Our contribution in these directions is twofold. We rigorously exhibit the claim about the expressiveness of $(\alpha,\beta)$-Dynamics, both by designing a simple protocol that provably computes the $k$-core of the network as well as by showing that $(\alpha,\beta)$-Dynamics are in fact Turing-Complete. Second and most important, we construct general tools for proving stabilization that work for a subclass of $(\alpha,\beta)$-Dynamics and prove speed of convergence in a restricted setting. 

\end{abstract}

\section{Introduction} \label{sec:intro}
\vspace{-0.2cm}
The interplay between the microscopic and the macroscopic in terms of emergent behavior shows an increasing interest. The most striking examples come from biological systems that seem to form macroscopic structures out of local interactions between simpler structures (e.g., computation of shortest paths by Physarum Polycephalum \cite{journals/n/NakagakiYT00}, or of maximal independent sets by the fly's nervous system \cite{ynoenz11}). The underlying common characteristic of these systems is the emergent behavior at the macroscopic level out of simple local interactions at the microscopic level. This is one of the reasons why in the last years there has been a surge in the analysis and design of elementary and fundamental primitives in distributed systems under restrictive assumptions on the model \cite{10.1145/3388392.3388403}.
In some of these examples, the dynamic processes are purely structural with respect to the network. These examples include network generation models \cite{Barabasi509,Watts1998Collective}, community detection \cite{DBLP:conf/kdd/ZhangWWZ09}, "life-like" cellular automata \cite{doi:10.1142/S0219525911003050}, robot motion \cite{saidani2004} and go all the way up to fundamental physics as a candidate model for space \cite{2002:NKS:513738,wolfram2020class}. In view of this recent trend, a stream of work is devoted to the study of such dynamics per se, without a particular application in mind (e.g., \cite{DBLP:journals/nc/Gadouleau20}).
Motivated by such a plethora of examples, we study the stabilization properties of protocols that affect solely the structure of networks. 

Henceforth, we will use the term \textit{dynamic network} to represent networks that change due to some process, although in the literature one can find other terms like adaptive networks, time-varying networks, evolving networks and temporal networks that essentially refer to the same general idea of time-dependent networks w.r.t. structure and states.
The study of the processes that drive dynamic networks and their resulting properties has been the focus of many different fields but in general one can discern between two distinct viewpoints without excluding overlap: \textbf{a) complex systems viewpoint (physics, sociology, ecology, etc.):} the main focus is on modeling (e.g., differential/difference equations, cellular automata, etc. - see \cite{Sayama2009}) and qualitative analysis (by means of mean field approximations, bifurcation analysis etc.). The main questions here are of qualitative nature and include phase transitions, complexity of system behavior, etc. Rigorous analysis is not usual and simulation is the main tool for providing results. \textbf{b) computational viewpoint (mainly computer science and communications):} the main focus is on the computational capabilities (computability/complexity) of dynamic networks in various settings and with different assumptions. The  main approach in computer science is based on rigorous proofs while in communications it is based on experimentation.

When designing local rules aiming at some particular global/emergent behavior, it is usually difficult, or at the very least cumbersome, to prove correctness \cite{10.1145/3388392.3388403}. This is why most studies in complex systems of this sort are based on experimental evidence for their correctness. 
Thus, it is very important to prove general results about protocols, and not argue about them in a case-by-case fashion. 
In this paper, we study a dynamic network driven by a simple protocol that is executed by each node in a synchronous manner. The protocol is the same for all nodes and can only affect the structure of the network and not the state of edges or nodes. The locality of the protocol is defined with respect to the available interactions for each node that are defined by a scheduler.
We define the \ab-Dynamics in Section~\ref{sec:preliminaries} and we also discuss related work. In Section~\ref{sec:min}, we discuss a particular protocol that computes the $\alpha$-core and the $(\alpha-1)$-crust \cite{DBLP:conf/gd/BatageljMZ99} of an arbitrary provided network. 
In Section~\ref{sec:degree} we provide guarantees on the speed of stabilization for a subclass of \abd while in Section~\ref{sec:Local_Rules} we provide a proof of stabilization for a more general class of such protocols. In this way, we provide general results for \abd that may be directly applied elsewhere, e.g., in the case of restricted Network Automata \cite{doi:10.1142/S0219525911003050}.  In Section~\ref{sec:turing} we prove that \abd is Turing-Complete. Finally, in Section~\ref{sec:extensions} we discuss some extensions of the proposed model and we conclude in Section~\ref{sec:conclusion}.

\section{Preliminaries} \label{sec:preliminaries}
\vspace{-0.2cm}
Assume that an undirected simple network $G^{(0)}=(V,E^{(0)})$ evolves over time (discrete time) based on a set of rules. We represent the network at time $t$ by $G^{(t)}=(V,E^{(t)})$. We denote the {\em distance} between two nodes $u,v$ in $G^{(t)}$ as $d^{(t)}(u,v)$. Let $n=|V|$, $m^{(t)}=|E^{(t)}|$ and let $N_{G^{(t)}}(u)$ be the set of all neighbors of node $u$ and $d_{G^{(t)}}(u)$ be the {\em degree} of node $u$ in network $G^{(t)}$. We define $\left| E^{(t)}(u,v) \right|$ to be the number of edges between $u$ and $v$ at time $t$ (either $0$ or $1$), and more generally $\left| E^{(t)}(U) \right|$ to be the number of edges between nodes in the set $U \subseteq V$ at time $t$. It follows that $\left| E^{(t)}(N_{G^{(t)}}(u) \cap N_{G^{(t)}}(v)) \right|$ is the number of edges between common neighbors of $u$ and $v$ at time $t$. Let $G^{(t)}[S]$ represent the induced subgraph of the node set $S \subseteq V$ in $G^{(t)}$.
The {\em potential} of a pair of nodes $u$ and $v$ at round $t$ is a function related to this pair and is represented by $\mathcal{E}_{G^{(t)}}^{(t)}(u,v):G^{(t)}[S] \rightarrow \mathbb{R}$, for some $S\subseteq V$. The domain of the potential is the induced subgraph $G^{(t)}[S]$ defined by the set of nodes $S$ that are at the local structure around nodes $u$ and $v$. This local structure is defined explicitly by the potential function. In this paper, $S$ consists of all nodes that are within constant distance from $u$ or from $v$ (the constant is $1$ throughout the paper, except for Section~\ref{sec:turing} where it is $3$). We write $\mathcal{E}^{(t)}(u,v)$ or $\mathcal{E}(u,v)$ when the network and the time we are referring to are clear from the context. An example of such a function defined in \cite{DBLP:conf/kdd/ZhangWWZ09} that is used to detect communities in networks is the following:
\[
\mathcal{E}(u,v)=| N_{G^{(t)}}(u) \cap N_{G^{(t)}}(v) | + | E^{(t)}(u,v) | + | E(N_{G^{(t)}}(u) \cap N_{G^{(t)}}(v)) |
\]
The potential is equal to the number of common neighbors between $u$ and $v$ plus the number of edges between $u$ and $v$ ($0$ or $1$) plus the number of edges between the common neighbors of $u$ and $v$. 

Finally, let $f: \mathbb{N}^2\rightarrow \mathbb{R}$ be a continuous function having the following two properties: i) Non-decreasing, that is $f(x,y+\epsilon)\geq f(x,y)$ for $\epsilon>0$ (similarly $f(x+\epsilon,y)\geq f(x,y)$) and ii) Symmetric, $f(x,y)=f(y,x)$. The second property is related to the fact that we consider undirected networks. We call these functions \textit{proper}.

\subsection{\abd~-~Thresholded Network Dynamics} \label{ssec:Network_System}
\vspace{-0.2cm}
Informally, the \ab-Thresholded Network Dynamics (\ab-Dynamics henceforth) in its general form is a discrete-time dynamic stateless network of agents $G^{(t)}=(V,E^{(t)})$. It is stateless because the dynamics driven by the protocol depend only on the structure of the network and not on state information stored in each node/edge. The dynamics involve the edges of the network while the set of agents is static. All interactions are pairwise and are defined by a scheduler. For each interaction, the two involved nodes execute a protocol that affects the edge between them. The execution of the protocol and all communication is carried out on the network $G^{(t)}$, while the scheduler is responsible for the determination of the interactions that activate the execution of the protocol between pairs of nodes in $G^{(t)}$. 

The protocol is \textit{consistent}, in the sense that it comes to the same decision about the existence of the edge between $u$ and $v$, both when executed by $u$ and by $v$. This requires the potential of an arbitrary edge $(u,v)$ to be \textit{computationally symmetric}, in the sense that $\mathcal{E}(u,v)$ is the same when computed in $u$ and in $v$. The execution evolves in synchronous discrete time rounds. In the following, the edge $e^{(t)}$ is also used as a boolean variable. In particular, when $e^{(t)}=0$ then $e^{(t)}\notin E^{(t)}$, while $e^{(t)}=1$ means that $e^{(t)}\in E^{(t)}$. Let $\alpha$ and $\beta$ be parameters that correspond to a lower and an upper threshold, respectively.
Initially, the network $G^{(0)}$ is given as well as the constant thresholds $\alpha$ and $\beta$. Formally, \abd is a triple $(G^{(0)},\mathcal{S},\mathcal{A}(\alpha,\beta))$ defined as follows:
\begin{itemize}
\item $G^{(0)}=(V,E^{(0)}):$ A network of nodes $V$ and edges $E^{(0)}$ between nodes at time $0$. This is the network where the dynamic process concerning the edges is performed. Each node $v\in V$ has a distinct id and maintains a routing table with all its edges. 

\item $\mathcal{S}:$ The scheduler that contains the pairwise interactions between nodes. We represent it by a possibly infinite series of sets of pairwise interactions $C^{(t)}$. Each set $C^{(t)}$ contains the pairwise interactions between nodes activated at time step $t$ in the network $G^{(t)}$. An interaction between nodes $u$ and $v$, assumes direct communication between $u$ and $v$ irrespective of whether $u$ and $v$ are connected by an edge in $G^{(t)}$. In the following, by slightly abusing notation, we will refer to $C^{(t)}$ as the scheduler for time step $t$. 

\item $\mathcal{A}(\alpha,\beta):$ The protocol executed in each round by each node participating in the pairwise interactions defined by the scheduler $C^{(t)}$ in order to update network $G^{(t)}$ to network $G^{(t+1)}$. The \abd is defined for the following family of protocols:
\begin{quote}
Protocol $\mathcal{A}(\alpha,\beta)$ at node $u$ for a pairwise interaction $(u,v)\in C^{(t)}$: 

\hspace{0.3cm} Compute the potential $\mathcal{E}(u,v)$.

\hspace{0.6cm} 1. If $\mathcal{E}(u,v)<\alpha$ then edge $(u,v)^{(t+1)}=0$.

\hspace{0.6cm} 2. If $\alpha \leq \mathcal{E}(u,v)<\beta$ then edge $(u,v)^{(t+1)}=(u,v)^{(t)}$. 
 
\hspace{0.6cm} 3. If $\mathcal{E}(u,v)\geq \beta$ then edge $(u,v)^{(t+1)}=1$. 

\end{quote}
\end{itemize}

The computational capabilities of each node are similar to a LOG-space Turing machine. Each node has two different memories, the input memory as well as the working memory. The input memory contains the local structural information of the network necessary for the computation of the potential function at node $u$. The potential function reads the input memory and its value is computed by using the working memory. We allow only protocols that require polynomial time w.r.t. the size of the input memory keeping the working memory logarithmic (asymptotically) in size w.r.t. the size of the input memory. 

The complexity of the protocol depends solely on the definition of the potential function, since the rest of the protocol are simple threshold comparisons. Similarly to dynamics \cite{10.1145/3388392.3388403} - although no relevant formal definition exists \cite{10.5555/3039686.3039745} - we require our protocol to be simple and lightweight and to realize natural, local and elementary rules subject to the constraint that structural dynamics are considered. To this end, we require the potential function to respect the following constraints: 
\begin{enumerate}
\item The potential function has access to a small constant distance $c$ away from the two interacting nodes.
\item The potential function must be indistinguishable with respect to the nodes - thus not allowing for special nodes (e.g., leaders) \cite{10.5555/3039686.3039745} \footnote{Therefore, we only use identifiers of nodes for analysis purposes}. 
\item The potential function must be network-agnostic, in the sense that it is designed without having any access to the topology of $G^{(0)}$.
\end{enumerate}
These restrictions combined with the computational capabilities of nodes do not allow the protocol to use shortcuts for computation in terms of hardwired information in the potential function (node ids) or in terms of replacing large subgraphs by other subgraphs. 

In each round, the protocol is executed by the nodes that participate in the pairwise interactions $(u,v)$ determined by the scheduler. A pairwise interaction between nodes $u$ and $v$ requires the computation of the potential between the two nodes and then a decision is made as for the edge between them based on the thresholds $\alpha$ and $\beta$. Each round of the computation for node $u$ (symmetrically for $v$) is divided into the following phases: (1) $u$ sends messages to its local neighborhood (with the exception of $v$, if edge $(u,v)$ exists) requesting information related to the computation of the potential function, (2) $u$ receives the requested information and stores it in the input memory, (3) $u$ sends its information to $v$, (4) $u$ receives $v$'s information and stores it in the input memory, (5) $u$ computes the potential using the working memory and (6) it decides as for the edge $(u,v)$ w.r.t. thresholds. 

The consistency of the protocol guarantees that the result of its execution is the same for $u$ and $v$. In accordance to the L\textsc{ocal} model, there is no restriction on the size of the messages. Finally, direct communication is assumed (in phases (3) and (4)) between the interacting nodes $u$ and $v$ irrespective of the existence of edge $(u,v)$.
In the example of the potential function given in Section~\ref{sec:preliminaries}, each round executes at $u$ (symmetrically for $v$) as follows: (1) $u$ sends messages to all its neighbors, (2) $u$ receives messages carrying information about its neighbors and their edges, (3) $u$ sends its gathered information to $v$, (4) $u$ receives the gathered information from $v$, (5) $u$ computes the potential between $u$ and $v$ and (6) it makes a decision about edge $(u,v)$ and appropriately updates its connection information. 

\abd is stateless, in the sense that the dynamics driven by the algorithm $\mathcal{A}$ consider only the structure of the network. No states that are stored at nodes or edges are considered in the dynamic evolution expressed by \ab-Dynamics. Although nodes have memory to store connections to their neighbors that change due to the dynamic process and to store the additional information required for the computation of the potential function, no additional states are used to impose changes in the network. 
As a result, the network $G^{(t)}$ completely defines the configuration of the system at time $t$.
We say that $G^{(t)}$ \textit{yields} $G^{(t+1)}$, when a transition takes place from $G^{(t)}$ to $G^{(t+1)}$ after time step $t$, represented as 
$G^{(t)} \xrightarrow{C^{(t)}}G^{(t+1)}$, which is the result of the $\mathcal{A}$ protocol for all pairwise interactions encoded in $C^{(t)}$.
Similarly, we write $G^{(t)} \rightsquigarrow G^{(t')}$, for $t'>t$, if there exists a sequence of transitions $G^{(t)}\xrightarrow{C^{(t)}} G^{(t+1)} \xrightarrow{C^{(t+1)}}\cdots\xrightarrow{C^{(t'-1)}}G^{(t')}$.
An \textit{execution} of \abd is a finite or infinite sequence of configurations $G^{(0)},G^{(1)},G^{(2)},\ldots$ such that for each $t$, $G^{(t)}$ yields $G^{(t+1)}$, where $G^{(0)}$ is the initial network. 

We say that the algorithm \textit{converges} or \textit{stabilizes} when $\exists t$ such that $\forall t'>t$ it holds that $G^{(t)}=G^{(t')}$, meaning that the network does not change after time $t$. The \textit{output} of the \abd is the network that results after stabilization has been reached.
The time complexity of the protocol is the number of steps until stabilization.
The time complexity of the protocol heavily depends on $C^{(t)}$. If, for example, there exists a $T$ where for all $t\geq T$ it holds that $C^{(t)}$ is always the null set, then the algorithm stabilizes although it would not stabilize for a different choice of $C^{(t)}$. To avoid stalling, we employ the \textit{weak fairness condition} \cite{10.1145/3289137.3289150,Angluin2006} that essentially states that all pairs of nodes interact infinitely often, thus imposing that the scheduler cannot avoid a possible change in the network.
In the case of the protocol described in Section~\ref{sec:min}, we will be very careful as to the definition of $C^{(t)}$ w.r.t. time complexity while for our stabilization theorems we either assume a particular $C^{(t)}$ or allow it to be arbitrary. However, in the latter case we do not claim bounds on the time complexity, only eventual stabilization. Note that it is not our goal in this paper to solve the problem of termination detection.

At this point, a discussion on the scheduler $\mathcal{S}$ is necessary. The scheduler $C^{(t)}$ at time $t$ supports parallelism since it is a set of pairwise interactions that has size at most $\binom{n}{2}$. Thus, many pairwise interactions may be activated in each step. For example, consider the case where all $\binom{n}{2}$ possible edges are contained in $C^{(t)}$. This means that simultaneously the potential is computed for all possible pairwise interactions and the edges are updated analogously. In~\cite{DBLP:conf/kdd/ZhangWWZ09}, a serialization of this case is used to detect communities in networks. In general, we may assume anything about the scheduler (adversarial, stochastic, etc.). Arguing about an arbitrary set of pairwise interactions for each $t$ is the most general case, since $\mathcal{A}$ can make no assumption at all about the pairwise interactions that will be activated within each round but the fairness condition must be employed in order to argue about stabilization.  

On a more technical note, the scheduler has two different but not necessarily mutually exclusive uses. On the one hand, the scheduler models restrictions set by the environment on the interactions (e.g., random interactions in a passive model). On the other hand, it is used as a tool for analysis reasons, to describe the communication links that the protocol $\mathcal{A}$ enforces on $G^{(t)}$ (e.g., when a node communicates with all nodes at distance $2$). The scheduler cannot and should not cheat, that is to be used in order to help $\mathcal{A}$ carry out the computation. In this paper, we present some general results w.r.t. the choice of the scheduler. For example, $C^{(t)}$ may be adversarial for all $t$, satisfying the fairness condition, while our algorithms are still able to stabilize (see Sections \ref{sec:min} and~\ref{sec:Local_Rules}).
Although \abd may seem to be a rather restricting setting, the freedom in defining the potential and the parameters $\alpha$ and $\beta$ allow us to have very rich behavior - in fact, we show that \abd is Turing-Complete.

\subsection{Related Work} \label{ssec:related}
\vspace{-0.2cm}
The main work on dynamic networks stems either from computer science or from complex systems and is inherently interdisciplinary in nature. In the following, we only highlight results that are directly related to ours (a more extensive discussion can be found in \cite{Michail2016}). In computer science, a nice review of the dynamic network domain \cite{DBLP:journals/cacm/MichailS18} proposes a partitioning of the current literature into three subareas: Population Protocols (\cite{Angluin2006}, \cite{DBLP:journals/dc/AngluinAER07}), Powerful Dynamic Distributed Systems (e.g., \cite{O'Dell:2005:IDH:1080810.1080828}) and models for Temporal Graphs (e.g., \cite{doi:10.1080/17445760.2012.668546}). 
\abd can be compared to Population Protocols, where anonymous agents with only a constant amount of memory available interact with each other and are able to compute functions, like leader election. Their scheduler determines the set of pairs of nodes among which one will be chosen for computation at each time step. The choice is made by a scheduler either arbitrarily (adversarial scheduler) or uniformly at random (uniform random scheduler). The uniform scheduler is used for designing various protocols due to the probabilistic accommodations for analysis it provides. The major differences to our approach are with respect to dynamics and the scheduler. Population protocols study state dynamics while in our case we study stateless structural dynamics. In addition, in our approach, the scheduler consists of a set of pairwise interactions, thus allowing for many computations between pairs of nodes during a time step (parallel time). This parallelism of the scheduler may "artificially" reduce the number of rounds but it can also complicate the protocol leading to interesting research questions. Similarly to population protocols, the notion of dynamics \cite{10.5555/3039686.3039745,10.1145/3388392.3388403} that refers to distributed processes that resemble interacting particle systems considers simple and lightweight protocols on states of agents. \abd could be cast in such a framework as purely structural dynamics that on the one hand supports simple, uniform and lightweight protocols while on the other hand requires necessarily the communication of structural information between nodes.
In the same manner, motivated by population protocols, the Network Constructors model also studies state dynamics that affect the structure of the network resulting in structural dynamics as well, and thus it is much closer to \ab-Dynamics. In \cite{Michail2016,Michail2017NetworkCA} the authors study what stable networks can be constructed (like paths, stars, and more complex networks) by a population of finite-automata. Among other complexity related results they also argue that the Network Constructors model is Turing-Complete. Our main differences to the network constructors model are the following: 
\begin{enumerate}
    \item Our motivation comes from the complex systems domain as well, and thus we are more interested in as general as possible convergence/stabilization theorems apart from particular network constructions (like the $\alpha$-core in our case).
    \item They use states for the structural dynamics while in our case the dynamics are stateless. This means that Network Constructors use states that change according to the protocol, which in turn drive the structural changes of the network (coupled dynamics). In our case, we use only the knowledge of the structure of the network to make structural changes. 
    \item They always start from a null network while we start from an arbitrary one. 
\end{enumerate}
A similar notion is graph relabeling systems \cite{graphrelabelling99}, where one chooses a subgraph and changes it based on certain rules. These systems are usually applied on static graphs but they have also been applied to dynamic graphs \cite{10.1007/978-3-642-11476-2_11}. The focus in this case is to \textit{impose properties on the dynamic graphs so that a particular computation is possible}, assuming adversarial dynamic graphs. 
\abd is also related - in fact can easily simulate - to graph generating models. The Barab{\'a}si--Albert model \cite{Barabasi509} can be simulated by simply setting $\mathcal{A}$ to add an edge between two nodes in $G^{(t)}$ for each interacting pair in $C^{(t)}$. These interacting pairs in $C^{(t)}$ are specified based on the stochastic preferential-attachment mechanism. Similarly, the Watts-Strogatz model \cite{Watts1998Collective} can be simulated by starting with a regular ring lattice and then in each step set the appropriate edges stochastically in $C^{(t)}$ to rewire them.

In the study of complex systems, one of the tools used for modeling is cellular automata. Cellular automata use simple update rules that give rise to interesting patterns \cite{DBLP:journals/corr/ArrighiD16}, \cite{DBLP:journals/corr/abs-1711-10920}. Structurally Dynamic Cellular Automata (SDCA) that couples the topology with the local site 0/1 value configuration were introduced in \cite{Ilachinski2018}. They formalize this notion and move to an experimental qualitative analysis of its behavior for various parameters. They left as an extension (among others) of SDCA purely structural CA models in which there are no value configurations as it holds in the \abd studied in this paper. A model for coupling topology with functional dynamics was given in \cite{doi:10.1142/S0219525911003050}, termed Functional Network Automata (FNA), and was used as a model for a biological process. They also defined the restricted Network Automata (rNA), which as \abd allows only for stateless structural network dynamics. rNA forces every possible pair of interactions to take place, meaning that for all $t$ it holds that $C^{(t)}$ contains all $\binom{n}{2}$ possible edges of the $n$ nodes. All their results are qualitative and are based on experimentation. By using the machinery built in Section~\ref{sec:Local_Rules} we show that for the family of protocols we consider, rNA always stabilizes.
To further stimulate the reader as for the need of looking at \ab-Dynamics, the author in \cite{saidani2004} looked at modular robots as an evolving network with respect only to their topology. The author defined a graph topodynamic, which in fact is a local program common to all modules of the robot, that turns a tree topology to a chain topology conjecturing that stabilization is always achieved but to the best of our knowledge it is still unresolved.

\section{Taking the Minimum} \label{sec:min}
\vspace{-0.2cm}
As a motivation and exhibition of \ab-Dynamics, we first discuss the following interesting example. We define the potential of a pair of nodes $u$ and $v$ as $\mathcal{E}(u,v)=\min\{d_{G^{(t)}}(u),d_{G^{(t)}}(v)\}$, that is the potential is equal to the minimum degree of the two nodes. This potential function respects all constraints described in \ref{ssec:Network_System}. 

It is interesting to notice the similarity of our process, and the process of acquiring the $k-core$ (or complementary the $(k-1)-crust$) of a simple undirected graph \cite{DBLP:conf/gd/BatageljMZ99,SzekeresW}.

\begin{definition}
The $k$-core $H$ of a graph $G$ is the unique maximal subgraph of $G$ such that $\forall u\in H$ it holds that $deg_H(u)\geq k$. All nodes not in $H$ form the $(k-1)$-crust of $G$.
\end{definition}
The $k$-core plays an important role in studying the clustering structure of networks \cite{mgpv20}. In \cite{DBLP:conf/gd/BatageljMZ99} it was proved that the following process efficiently computes the $k$-core of a graph:
\begin{lemma} \label{lem:a-core}
Given a graph $G$ and a number $k$, one can compute $G$'s $k$-core by repeatedly deleting all nodes whose degree is less than $k$.
\end{lemma}

The following theorem states that stabilization to the $k$-core is achieved for an arbitrary scheduler $\mathcal{S}$. Furthermore, the stabilization occurs after $O(m)$ rounds of changes in the network, where $m$ is the number of edges in $G$. Note that this is not the time complexity of the protocol, since there may be many idle rounds between rounds with changes, depending on the scheduler.

\begin{theorem}
When $\mathcal{E}(u,v)=\min\{d_{G^{(t)}}(u),d_{G^{(t)}}(v)\}$, \abd for any value of $\alpha\leq n-1 < \beta$ and any scheduler $\mathcal{S}$, stabilizes in a network where all isolated nodes form the $(\alpha-1)$-crust and the rest the $\alpha$-core of $G^{(0)}$ in $O(m)$ rounds where changes happen, where $m$ is the number of edges in $G^{(0)}$.
\end{theorem}
\begin{proof}
First of all, even if a node connects with any other node, its degree will be $n-1$. Thus, it holds that $min\{d(u),d(v)\}\leq n-1 < \beta$. This ensures that no edge will ever be created by the \ab-Dynamics. Thus, only deletions of edges can be performed. As a result, the maximum number of rounds where a change happens is a straightforward $O(m)$. What we need to show is that the output of the protocol is a network where all isolated nodes belong to the $(\alpha-1)$-crust of $G^{(0)}$ and the rest of the nodes belong to the $\alpha$-core of $G^{(0)}$.

To prove our claim we change slightly the algorithm described in Lemma~\ref{lem:a-core} to process edges instead of nodes. This change is made so that the \abd described in this section will be in fact a realization of this main memory algorithm and thus its output will be the $\alpha$-core of $G^{(0)}$.
Indeed, one can compute $G$'s $\alpha$-core by repeatedly deleting all edges for which one of its endpoints has degree $<\alpha$. The procedure stops when there is no such remaining edge, that is, all edges have endpoints with degree $\geq \alpha$. The order in which the edges are considered is irrelevant. It is easy to see that this algorithm computes the $\alpha$-core of the given network and in fact it is the \abd described in this section.
\end{proof}

A final note concerns the time complexity. Note that the aforementioned theorem does not state anything about the time complexity of the protocol, it just states the maximum number of rounds where a change happens. We can compute the time complexity if we describe the scheduler. If we assume that $\forall t: C^{(t)}=E^{(t)}$, that is the scheduler contains all edges and only those of the $G^{(t)}$ network then the time complexity is $O(n)$. This is because, at each round it is guaranteed that one node will become isolated unless stabilization has been achieved. Similarly, if we assume a uniform scheduler that chooses one pair of nodes uniformly at random in each time step, then the \abd stabilizes in $O(mn^2\log{m})$ steps by a simple application of the coupon collector problem on the selection of edges.

\section{\abd~with $\alpha=\beta$~~and a Proper Potential Function on the Degrees} \label{sec:degree}
\vspace{-0.2cm}
We study the \abd where the potential is any symmetric non-decreasing function on the degrees of its two endpoints. We prove that in this case \abd  stabilizes while the time complexity is $O(n)$, assuming that $\alpha=\beta$ and that for all $t$, $C^{(t)}$ contains all $\binom{n}{2}$ possible pairwise interactions. All proofs can be found in
Appendix~\ref{app:degree}.
More formally, we define the potential of a pair $(u,v)$ to be $\mathcal{E}(u,v)=f(d_{G^{(t)}}(u),d_{G^{(t)}}(v))$, where $f$ is a $\textit{proper}$ (symmetric and non-decreasing in both variables) function. Since $f$ is proper, the potential function is computationally symmetric and thus the protocol is consistent. 

For the network $G^{(t)}$, let $R^{(t)}(u,v)$ be an equivalence relation defined on the set of nodes $V$ for time $t$, such that $(u,v)\in R^{(t)}$ iff $d_{G^{(t)}}(u)=d_{G^{(t)}}(v)$. The equivalence class $R^{(t)}_i$ corresponds to all nodes with degree $d(R^{(t)}_i)$, where $i$ is the rank of the degree in decreasing order. Thus the equivalence class $R^{(t)}_1$ contains all nodes with maximum degree in $G^{(t)}$. Assuming that $n=|V|$, the maximum number of equivalence classes is $n-1$, as the degree can be in the range $[0,n-1]$ and no pair of nodes $(u,v)$ with degrees $d_{G^{(t)}}(u)=0$ and $d_{G^{(t)}}(v)=n-1$ can exist. Let $|G^{(t)}|$ be the number of equivalence classes in $G^{(t)}$. 

We prove by induction that in this setting, \abd always stabilizes in at most $|G^{(0)}|+1$ steps. To begin with, the clique $\mathcal{K}_n$ as well as the null graph $\overline{\mathcal{K}_n}$ both stabilize in at most one step, for any value of $\beta$. The following renormalization lemma describes how the number of equivalence classes is reduced and is crucial to the induction proof. 

\begin{lemma}
If $d(R^{(t)}_1) = n-1$, $\forall t\geq c, c\in \mathbb{N}$, and the subgraph $G^{(c)} \setminus R^{(c)}_1$ stabilizes for any value of $\beta$ and proper function $f$, then $G^{(c)}$ stabilizes as well. Similarly, if $d(R^{(t)}_{|G^{(t)}|}) = 0$, $\forall t\geq c, c\in \mathbb{N}$, and the subgraph $G^{(c)} \setminus R^{(c)}_{|G^{(c)}|}$ stabilizes for any value of $\beta$ and proper function $f$, then $G^{(c)}$ stabilizes as well. The time it takes for $G^{(c)}$ to stabilize is the same as the time it takes for the induced subgraph to stabilize for both cases.
\end{lemma}
The following theorem establishes stabilization in linear time. 

\begin{theorem}
When $\alpha=\beta$, $f$ is proper, $\mathcal{E}(u,v)=f(d_{G^{(t)}}(u),d_{G^{(t)}}(v))$, and the scheduler contains all $\binom{n}{2}$ possible pairwise interactions in each time step, \abd with input $G^{(0)}$ stabilizes in at most $|G^{(0)}|+1$ steps.
\end{theorem}

\section{\abd~Stabilization for Arbitrary Scheduler} \label{sec:Local_Rules}
\vspace{-0.2cm}
In this section, we prove stabilization (with no speed bound) for any $\alpha\leq \beta$ in an adversarial setting where the scheduler $\mathcal{S}$ may be completely arbitrary subject to the fairness condition.
In addition, we further generalize by changing the definition of potential, from $\mathcal{E}(u,v)=f(d_{G^{(t)}}(u),d_{G^{(t)}}(v))$ to $\mathcal{E}(u,v)=f(g_{G^{(t)}}(u),g_{G^{(t)}}(v))$, for a family of functions $g_{G}:\mathbb{R}^k\rightarrow \mathbb{R}, k\in\mathbb{N}$.
We call a function $g_G(u)$ \textit{degree-like} if it only depends on the neighborhood $N_G(u)$ of node $u$ and has the following property: assuming that the neighborhood of node $u$ at time $t$ is $N_{G^{(t)}}(u)$, and the neighborhood of $v$ at time $t'$ is $N_{G^{(t')}}(v)$, and $N_{G^{(t)}}(u) \supseteq N_{G^{(t')}}(v)$, then we require that $g_{G^{(t)}}(u) \geq g_{G^{(t')}}(v)$. The reason we extend the notion of degree is to represent more interesting rules as shown in the toy model of social dynamics of Section~\ref{sec:extensions}.

The potential function is computationally symmetric since $f$ is proper and $g$ is common for $u$ and $v$. The protocol in Section~\ref{sec:degree} is a special case of this protocol, where $g$ is the degree of the node, the scheduler contains all $\binom{n}{2}$ possible pairwise interactions at each time step and $\alpha=\beta$.
We first need the following definition:

\begin{definition}
A pair $(t,D)$ is $|D|-Done$ if $t\in \mathbb{N}$, $D \subseteq V$ and  $\forall u \in D$ it holds that their neighborhood does not change after time $t$. That is, $N_{G^{(t')}}(u) = N_{G^{(t)}}(u)$, for $t'\geq t$.
\end{definition}
Our stabilization proof repeatedly detects $|D|-Done$ pairs with increasing $|D|$. When $D=V$, all neighborhoods do not change, and thus the process stabilizes.

\begin{lemma}\label{lem:Increasing_D_Done}
If there exists a $|D|-Done$ pair $(t,D)$ at round $t$ with $|D|<|V|$, then $\exists t'>t$ such that at round $t'$ there exists a $(|D|+1)-Done$ pair $(t',D')$.
\end{lemma}
\begin{proof}

The core idea is to find a time-step $t_1$ where a node $u\not \in D$ maximizes $g$, as specified in the next paragraph; if $u$ never drops any edge in subsequent time steps, we prove that its neighborhood is stabilized, and we extend $D$ by $u$; if it drops an edge with a node $w$, this node $w$ is not able to preserve any other edge, due to the selection of $u$, and we are able to extend $D$ by $w$.

More formally, we denote by $t_1\geq t$ the time-step where there is some node $u \not\in D$ such that $g_{G^{(t_1)}}(u)\geq g_{G^{(t_1')}}(v)$, for all $t_1'\geq t_1$ and $v \not \in D$. If there are many choices for $t_1$ and $u$, we pick any $t_1$ and $u$ such that $u$ has the highest degree possible. Note that, later in time (say at $t_1'>t_1$), it is entirely possible that $u$'s neighborhood shrinks and thus its $g$ value drops $(g_{G^{(t_1')}}(u) < g_{G^{(t_1)}}(u))$. 
It is guaranteed that $t_1$ exists, as there are finitely many graphs with $|V|$ nodes, and finitely many nodes. Thus, there are finitely many values of $g_{G}(u)$ to appear after time $t$. Additionally, the fairness condition guarantees that the pairwise interaction between $u$ and $v$ will be eventually activated, for any $v$. 

If $u$ never drops any edge after $t_1$, then its neighborhood can only grow or stay the same. But if its neighborhood grows, due to the properties of function $g$, its value will not drop and the degree of $u$ will increase. However, the way we picked $u$ and $t_1$ does not allow this. We conclude that the neighborhood of $u$ does not change after time $t_1$, and thus we can extend $D$ by $\{u\}$, that is $(t_1,D \cup \{u\})$ is $(|D|+1)-Done$.
Else, if $u$ drops an edge after $t_1$, let $t_2>t_1$ be the first time step that a neighbor $w$ of $u$ in $G^{(t_2-1)}$ is not a neighbor of $u$ in $G^{(t_2)}$. Since $u$'s neighborhood stays the same until $t_2-1$, it follows that $g_{G^{(t_1)}}(u)=g_{G^{(t_2-1)}}(u)$.
The neighborhood of $w$ does not grow at subsequent time steps, that is $N_{G^{(t_2')}}(w) \supseteq N_{G^{(t_2'+1)}}(w)$, $t_2'\geq t_2-1$. To prove this, we show that $w$ never forms a new edge after $t_2-1$. Suppose it does at $t_2'+1$ for the first time. Then $w$ forms an edge with some node $v \not \in D$, due to the definition of $D$. However, we know that $\beta\geq \alpha > f(g_{G^{(t_2-1)}}(u),g_{G^{(t_2-1)}}(w)) = f(g_{G^{(t_1)}}(u),g_{G^{(t_2-1)}}(w)) \geq f(g_{G^{(t_2')}}(v),g_{G^{(t_2')}}(w))$, due to $f$ being non-decreasing, $g$ being degree-like, and the definition of $u$ and $t_1$. Thus, an edge between $v$ and $w$ cannot be formed.

We conclude that the neighborhood of $w$ can only shrink after time $t_2$. But there are only finitely many options for the neighborhood of $w$, and thus there is a time $t_3\geq t_2$ where the neighborhood of $w$ is the same in all subsequent graphs. Therefore, we can extend $D$ by $\{w\}$, that is $(t_3,D \cup \{w\})$ is $(|D|+1)-Done$.
\end{proof}

\begin{theorem} \label{thm:conv_only}
For $\mathcal{E}(u,v)=f(g_{G^{(t)}}(u),g_{G^{(t)}}(v))$, \abd stabilizes for any $\alpha\leq \beta$, proper function $f$, degree-like function $g$ and arbitrary scheduler $\mathcal{S}$ subject to the fairness condition.
\end{theorem}
\begin{proof}

It trivially holds that $(0,\emptyset)$ is $0-Done$. By applying Lemma~\ref{lem:Increasing_D_Done} once, we increase the size of $D$ by $1$. Thus, by applying it $|V|$ times, we end up with a $|V|-Done$ pair $(t,V)$. Since all neighborhoods stay the same for all future steps, $G^{(t')}=G^{(t)}$ for all $t'\geq t$.
\end{proof}
Theorem~\ref{thm:conv_only} can directly prove stabilization of the protocol in Section~\ref{sec:min}.

\section{Turing-Completeness}\label{sec:turing}
\vspace{-0.2cm}
In this section we describe the \abd that is able to simulate Rule $110$, a one-dimensional Cellular Automaton (CA) that Cook proved to be Turing-Complete \cite{DBLP:journals/compsys/000104} (for a discussion on CA and Rule 110, see
Appendix~\ref{app:ssec_ca-rule110}).
Thus, we prove that \abd is Turing-Complete as well, meaning that it is computationally universal since it can simulate any Turing machine (or in other terms any algorithm).
All proofs of theorems and lemmas in this section can be found in Appendix~\ref{app:ssec_tc-proofs}.

\begin{definition}
Rule $110$ is a one-dimensional CA. Let $cell^{(t)}(i)$ be the binary value of the $i$-th cell at time $t$. If $cell^{(t)}(i)=0$, then $cell^{(t+1)}(i)=cell^{(t)}(i+1)$. Else, $cell^{(t+1)}(i)$ is $0$ if $cell^{(t)}(i-1)=cell^{(t)}(i+1)=1$, and $1$ otherwise.
\end{definition}

Let $CN^{(t)}(u,v) = |N_{G^{(t)}}(u) \cap N_{G^{(t)}}(v)|$ be the number of common neighbors of $u$ and $v$ at time $t$, and $CE^{(t)}(u,v)=\left| E(CN^{(t)}(u,v)) \right|$ be the number of edges between the common neighbors of $u$ and $v$ at time $t$. For the following simulation we assume w.l.o.g. that $\alpha=\beta$ and that the scheduler $\mathcal{S}$ contains all possible $\binom{n}{2}$ interactions, for all time steps. The potential between nodes $u$ and $v$ is defined as follows:

\[   
\mathcal{E}^{(t)}(u,v) = 
     \begin{cases}
       \beta+60+CE^{(t)}(u,v)-CN^{(t)}(u,v)	& if~ 66\le CN^{(t)}(u,v)+|E^{(t)}(u,v)|\le 70	\\
       \beta+12-CE^{(t)}(u,v)	& if~ CN^{(t)}(u,v)+|E^{(t)}(u,v)|=71	\\
       \beta - |E^{(t)}(u,v)|   & if~ 40\le CN^{(t)}(u,v)\le 41  							\\
       
       \beta-1+|E^{(t)}(u,v)|	& otherwise
     \end{cases}
\]

The first $2$ branches are the ones that are actually related to Rule $110$, and are used only in Lemma~\ref{lem:Simulate_110}. The rest of them are only used in Lemma~\ref{lem:tc-main-structure} and ensure technical details, namely that some pairs of nodes always flip the status of their connection (Branch $3$), effectively providing us with a clock, and some of them always preserve it (Branch $4$). 

As required, computing the function only uses a constant number of words in the working memory, which have logarithmic size in bits compared to the input memory (which contains the neighborhoods of $u$ and $v$), and requires polynomial time in the size of the input memory. For example, to compute $CN^{(t)}(u,v)$, one could iterate over all pairs $(u',v')$ such that $u\in N_{G^{(t)}}(u), v\in N_{G^{(t)}}(v)$, and increment a counter initially set to zero, every time $u'=v'$. Similarly, to compute $CE^{(t)}(u,v)$, one can iterate over quadruples $u',u'',v',v''$ and increment a counter whenever $u'=v', u''=v''$ and there exists an edge between $u'$ and $u''$. Additionally, the potential function only depends on nodes at a constant distance (at most $1$) from either $u$ or $v$, and it is network-agnostic (not assuming access on the topology of $G^{(0)}$). Finally it is computationally symmetric and thus the protocol is consistent.

Informally, our simulation of Rule $110$ consists of the following steps. First, we design a primitive cell-gadget (henceforth $PCG$) that stores binary values, but fails to capture Rule $110$ since it doesn't distinguish between the left and the right cell. Then, by making use of the $PCG$ as a building block, we build the main cell-gadget (henceforth $CG$) that is used to simulate a single cell of the CA. Then, each time step from Rule $110$ is simulated using $2$ rounds of the \ab-Dynamics; on the first round, some $PCG$s acquire their proper value while on the second round, the rest of the $PCG$s copy the correct value from the ones that already acquired it. Finally, the two steps are merged into one in order to achieve stabilization of the dynamics when Rule $110$ has also stabilized.

For clarity purposes, we slightly abuse notation, and we count the rounds of the \abd by multiples of $0.5$ instead of $1$. Thus, we write that the sequence of configurations is $G^{(0)}, G^{(0.5)}, G^{(1)}...$, where configurations $G^{(t+0.5)}$, for $t \in \mathbb{N}$, are transitional states of the network and have no correspondence with cell states of the CA.

In order to construct the $PCG$ and the $CG$, we first construct two auxiliary gadgets, the always-on $(x,y)$-gadget and the flip $(x,y)$-gadget. The always-on $(x,y)$-gadget is simply a clique of $22$ nodes. $20$ of them have no edges to other nodes in the network, while $2$ of them (namely $x$ and $y$) may be connected with other nodes. The flip $(x,y)$-gadget is basically two always-on $(x,y)$-gadgets, with nodes $x$ and $y$ being the same for both gadgets, with the exception that the edge between $x$ and $y$ may not exist. See Figure~\ref{fig:gadgets} for both of these gadgets. We later show that, under certain conditions, the edge between $x$ and $y$ always exists in an always-on gadget, and flips its state at each time step, in a flip gadget.

A $PCG$ consists of a pair of nodes $(h,l)$, such that the existence of an edge between them corresponds to value $1$ and otherwise it corresponds to value $0$, and $60$ auxiliary nodes $a_{1}, \ldots a_{60}$. Furthermore, for each of the $120$ pairs of the form $(h,a_{i})$ and $(l,a_{i})$, there exists a corresponding $(h,a_{i})$ and $(l,a_{i})-$flip gadget. When we have two different $PCG$s, say $A$ and $B$, we write $A(h), A(l), A(a_1), \ldots , A(a_{60})$ for the nodes of $A$ and similarly $B(h), B(l), B(a_1),$ $\ldots , B(a_{60})$ for the nodes of $B$. We write $A^{(t)}$ to denote the value of $A$ at time $t$; in other words $A^{(t)} = |E^{(t)}(A(h),A(l))|$.

In order to connect two different $PCG$s (say $A$ and $B$) we add $4$ always-on gadgets: the always-on $(A(h),B(h))$ gadget, the always-on $(A(h),B(l))$ gadget, the always-on $(A(l),B(h))$ gadget and the always-on $(A(l),B(l))$ gadget, as shown in Figure~\ref{fig:gadgets}. Intuitively, this relates $CE^{(t)}(A(h),A(l))$ to the sum of values of the connected $PCG$s.

\begin{figure}[t]
\centering
\includegraphics[scale=0.7]{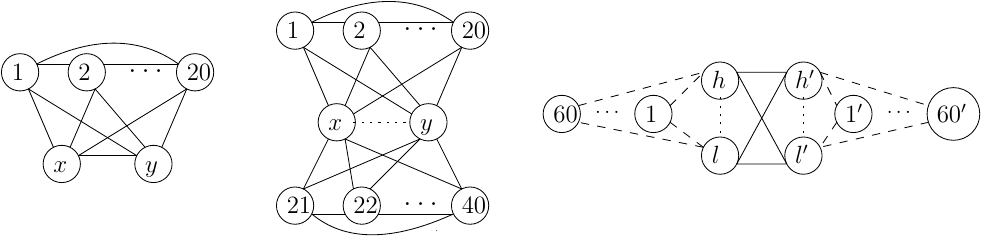}
\caption{To the left, we have an always-on $(x,y)$ gadget. In the middle, we have a flip $(x,y)$ gadget; the dotted line between $(x,y)$ denotes that this particular edge may or may not exist. To the right, we have two $PCG$s. The dashed lines denote flip gadgets, the dotted lines denote that these particular edges may or may not exist. The continuous lines denote always-on gadgets; these $4$ always-on gadgets is how we connect $PCG$s.}
\label{fig:gadgets}
\end{figure}

The $i$-th $CG$ that corresponds to the $i$-th cell (we write $CG(i)$) consists of $4$ $PCG$s, which we identify as $A_1(i)$, $A_2(i)$, $B_1(i)$ and $B_2(i)$. At time $t=0$, the edge in each flip gadget of $A_1(i), A_2(i)$ exists, while the edge in each flip gadget of $B_1(i), B_2(i)$ does not exist. We connect each $A_j(i)$ with each $B_k(i)$ ($4$ connections in total, where each connection uses $4$ always-on gadgets, as depicted in Figure~\ref{fig:gadgets}). In order to connect $CG(i)$ (cell $i$) with $CG(i+1)$ (cell $i+1$) we connect $A_j(i)$ with $A_j(i+1)$, and $A_j(i)$ with $B_j(i+1)$.
A $CG$ is said to have value $0$ if all $4$ of its $PCG$s are set to $0$ and $1$ if all $PCG$s are set to $1$. We guarantee that no other case can occur in $G^{(t)}, t\in \mathbb{N}$, although this is not guaranteed for the intermediate configurations $G^{(t+0.5)}, t\in \mathbb{N}$.


To conclude the construction of $G^{(0)}$, each cell of Rule $110$ corresponds to a $CG$ in $G^{(0)}$, and neighboring cells have their corresponding $CG$s connected. Finally, we set the value of its $CG$ (that is the value of its $4$ $PCG$s) equal to the initial value of the corresponding cell. 

Notice that all our gadgets are defined for a single time-step, namely for $t=0$. One could imagine that in subsequent time-steps, nodes contained in the same gadget in $G^{(0)}$ are no longer connected in the same way (effectively destroying the gadget), or even that new gadgets are formed. The following lemma shows that this is not the case. Informally, it shows that no new gadgets are created, and that the only difference between graphs at different time steps concern edges that do not destroy the existing gadgets. For example, in the definition of a flip gadget, there is only one pair of nodes (its two special nodes) for which it does not matter whether they share an edge or not; the lemma shows that between nodes that belonged in the same flip gadget in $G^{(0)}$, only this special pair may change its connection (existence or not of an edge between them) through time.

\begin{lemma} \label{lem:tc-main-structure}

If there exists a flip $(x,y)$-gadget connected to an $A_j(i)$ $PCG$ in $G^{(0)}$, then the edge $(x,y)$ at time $t$ exists if and only if $t \in \mathbb{N}\cup \{0\}$. Similarly, if there exists a flip $(x,y)$-gadget connected to a $B_j(i)$ $PCG$ in $G^{(0)}$, then the edge $(x,y)$ exists if and only if $t \not \in \mathbb{N}\cup \{0\}$. Finally, all other edges exist at any time step if and only if they exist in $G^{(0)}$, with the exception of edges between $(h,l)$ nodes of a $PCG$.
\end{lemma}

Our next step is to discuss how $(h,l)$ edges of $PCG$s change. The number of common neighbors of an $h,l$ pair of an $A_j(i)$ is $CN^{(t)}(h,l)=70$, for all integer time steps $t$ and valid $i,j$, as it has $5$ neighboring $PCG$s (each contributing $2$), and $60$ auxiliary nodes within the PCG (by Lemma~\ref{lem:tc-main-structure}). For non-integer time steps $t+0.5, t\in \mathbb{N}\cup \{0\}$, by Lemma~\ref{lem:tc-main-structure}, the $60$ auxiliary nodes are not connected with $h$ and $l$, and so $CN^{(t)}(h,l)=10$. Similarly, the number of common neighbors of an $(h,l)$ pair of a $B_j(i)$ is $CN^{(t)}(h,l)=66$, for all non-integer $t$ and valid $i,j$, and $CN^{(t)}(h,l)=6$ for integer $t$.

Furthermore, for all $t$, it holds that $CE^{(t)}(A_j(i)(h), A_j(i)(l))=8+A_j^{(t)}(i-1)+B_1^{(t)}(i)+B_2^{(t)}(i)+A_j^{(t)}(i+1)+B_j^{(t)}(i+1)$, as the edges between common neighbors are the internal edges of connected $PCG$s, plus the connection between $A_j^{(t)}(i-1)$ and $B_j^{(t)}(i)$ ($4$ edges), plus the connection between $A_j^{(t)}(i+1)$ and $B_j^{(t)}(i+1)$ ($4$ edges). Similarly, for a $B_j(i)$ we have that $CE^{(t)}(B_j(i))=4+A_j^{(t)}(i-1)+A_1^{(t)}(i)+A_2^{(t)}(i)$.

\begin{lemma} \label{lem:Simulate_110}
It holds that $A_j^{(t)}(i)=B_j^{(t)}(i)=cell^{(t)}(i)$ for $j\in\{1,2\}$ and all $i,t\in \mathbb{N}$.
\end{lemma}

The following corollary is a straightforward consequence of this lemma.

\begin{corollary} \label{cor:eq}
It holds that $cell^{(t)}(i)=CG^{(t)}(i)$.
\end{corollary}

The above construction simulates Rule $110$. The only problem is that it takes two time steps to simulate a single time step of Rule $110$, meaning that even if Rule $110$ converges, our construction infinitely flips between two different configurations, due to the flip gadgets, and as a result it does not stabilize. To overcome this problem, we use the aforementioned construction and make changes that allow us to remove the intermediate steps in the simulation, that is the steps $t+0.5, t\in \mathbb{N}\cup \{0\}$.
\begin{theorem} \label{thm:tc}
The \abd is Turing-Complete.
\end{theorem}

\section{Extensions} \label{sec:extensions}
\vspace{-0.2cm}
We briefly discuss two straightforward extensions of \abd and provide related examples. To begin with, we can add static information to nodes/edges (e.g., weights). This information is encoded by the potential function and does not change with time. The degree-like function defined in Section~\ref{sec:Local_Rules} can be used to assign a time-independent importance factor (e.g. a known centrality measure in $G^{(0)}$) while letting $g(u)$ be the sum of these factors of nodes in $N_{G^{(t)}}(u)$.
To demonstrate it, we provide a small example with a toy model inspired by Structural Balance Theory \cite{10.2307/2572978} of networks with friendship and enmity relations \cite{ANTAL2006130}. This example is more reminiscent of population dynamics rather than distributed protocols. Assume that the network of agents corresponds to people (nodes) with friendship relations (edges). Each agent $v$ is defined by how nice she is $n(v)$, how extrovert she is $x(v)$ as well as by the set of her enemies $\mathcal{EN}(v)$. We wish to design a model that captures how friendships change in this setting when enemies do not change\footnote{The permanence of enmity is in fact not exactly compatible with structural balance theory on networks.} as well as when friendships are lost in case of very few common friends, while friends are made in the opposite case.

To define the social dynamics we need to define the scheduler and the potential function that essentially describe our toy model. The scheduler captures the interactions between the agents enforced by the model. This toy model is only for the purpose of highlighting our convergence results and we do not claim to realistically capture certain social phenomena. The scheduler is defined as follows:  (a) if two agents $u$ and $v$ are enemies then they never become friends (no pairwise interaction between them in $C^{(t)}$, for any $t$), (b) if two agents $u$ and $v$ are not connected by an edge in $G^{(t)}$ (they are not friends) but their distance is at most the sum of their extrovertedness, then they interact - that is, if at time $t$ it holds that $1<dist(u,v)\leq x(u)+x(v)$ then there is an edge $(u,v)$ in $C^{(t)}$, (c) if two agents are connected by an edge in $G^{(t)}$, then there is a pairwise interaction between them in $C^{(t)}$ if their number of common friends is $\leq \gamma$. If their common friends are $>\gamma$ then their friendship is strong and it will not be affected at this round, and thus no edge in $C^{(t)}$ is introduced. This concludes the description of the scheduler. 

As for the potential function, we define the potential between $u$ and $v$ in $G^{(t)}$ to be $\mathcal{E}(u,v)=(n(u)+\sum_{w\in N(u)}{n(w)})+(n(v)+\sum_{w\in N(v)}{n(w)})$, capturing our intuition that friendships are created or stopped based on how nice the two agents and their neighbors are. 
This is a computationally symmetric function and thus the protocol is consistent. The function $g$ corresponds to the sum of the niceness of a node plus the niceness of its neighbors and thus it is degree-like. The function $f$ is proper since it is a simple sum between $u$ and $v$ w.r.t. the output of the function $g$ in each node.  
Thus, \abd on this social network stabilizes by Theorem~\ref{thm:conv_only} (the proof holds without any modification, even in this somewhat extended version of \ab-Dynamics). Theorem~\ref{thm:conv_only} also allows us to add any rules w.r.t. the scheduler $\mathcal{S}$ like imposing a maximum number of friends, allowing for additional random connections (to achieve long-range interaction), etc. Similarly, we can change the definition of potential and still prove stabilization as long as the assumptions of Theorem~\ref{thm:conv_only} are valid. If these assumptions are violated, as it would be in the case of a potential function that applies to a subset of neighbors (e.g., common neighbors between $u$ and $v$), then a new analysis is required to prove stabilization, if stabilization can be reached.
Finally, the scheduler allows us to remove the assumption of permanence on enmity by allowing under certain conditions particular pairwise interactions, thus dynamically changing the set $\mathcal{EN}(v)$. 

Another straightforward generalization is to allow for general stateless protocols $\mathcal{A}$ targeting at providing algorithmic solutions for specific problems. An example of such a generalization is given below for constructing a spanning star. We show in simple terms the stateless approach when compared to state-dependent approaches for constructing a network (e.g., Network Constructors model \cite{Michail2016,Michail2017NetworkCA}). 
In some sense, we already provide such an example of explicit network construction in the case of the $\alpha$-core. We assume a uniform random scheduler, that is, in our model we assume that in each time step a pairwise interaction is chosen uniformly at random. In \cite{Michail2016} they provide a simple protocol that uses states on the nodes, which, starting from the null graph, constructs the spanning star in optimal $\Theta(n^2\log{n})$ expected time. We discuss a protocol $\mathcal{A}$ that computes a spanning star starting from any network. It is reminiscent of the random copying method \cite{10.1145/335168.335170} for generating power law networks. It would be interesting to find out whether hub-and-spoke networks (essentially star networks) can be generated by some similar social process. In this case, the probability of choosing pairwise interactions should be related to the degree of the involved nodes, leading to the definition of a non-uniform random scheduler.

To describe the protocol let $u$ and $v$ be two nodes that interact at time $t$ as determined by the scheduler. If no edge exists between them, an edge $(u,v)$ is added. Assume w.l.o.g. that $d_G^{(t)}(u)> d_G^{(t)}(v)$. Then, the protocol dictates that all edges of $v$ are to be moved to $u$. In case $d_G^{(t)}(u)=d_G^{(t)}(v)\neq 1$, we break symmetry (symmetry breaking was also needed in \cite{Michail2016} by the scheduler) by tossing a fair coin in each node as to which node is going to transfer its neighbors. The nodes communicate the result of their toss and if found equal no change happens in the current round, otherwise we again move all edges from the one node to the other. If $d_G^{(t)}(u)=d_G^{(t)}(v)=1$ then let $x$ and $y$ be the only neighbors of $u$ and $v$ respectively. If $d_G^{(t)}(x)=d_G^{(t)}(y)=1$, $x$ and $y$ toss a fair coin and if it happens to be different one of these nodes will be the root of a tree with three leaves. Otherwise, the same process is applied on $x$ and $y$ as in $u$ and $v$. Note that in this case the degrees of $x$ and $y$ cannot be both equal to $1$. 

On the positive side, the difference of this protocol to the one given in \cite{Michail2016} is that no state dynamics are used and we start from an arbitrary network. However, on the negative side, a pairwise interaction between $u$ and $v$ may affect all nodes up to distance $2$ since no states are used that could allow us to move these edges incrementally in future interactions. Correctness is proved based on the observation that in each round when a leaf node has its degree increased then the connected components of the network are reduced, otherwise either a node becomes a leaf or nothing happens due to the symmetry breaking mechanism. Because of this stalling due to symmetry breaking, the time complexity analysis is more involved but we conjecture only by a polylogarithmic factor away from the one in \cite{Michail2016} (due to moving the edges). The protocol could be simplified in order to change only the neighborhood of $u$ and $v$, but the time complexity would increase substantially. To exploit parallel time, we could allow for more interactions per round as long as those are not affecting each other.

\section{Conclusion}\label{sec:conclusion}
\vspace{-0.2cm}
\abd are stateless structural dynamics of a network. The protocol allows for two thresholds that affect the existence of the edges in the pairwise interactions determined by the scheduler at each time step. Since the dynamics are purely structural, the output of the protocol is another network, and thus \abd can be considered as a network transformation process. Such a process for example has been used in \cite{DBLP:conf/kdd/ZhangWWZ09} to detect communities. In fact, the authors wondered whether conditional convergence could be proved. It is a matter of technical details to show that for regular networks one can choose $\alpha$ and $\beta$ such that the protocol never stabilizes. 

For future research, it would be very interesting to look at the notion of parallel time in \ab-Dynamics. Another interesting research direction is to see the effect of higher order structural interactions as well as look at how the model is affected when messages are restricted in size (in accordance to the C\textsc{ongest} model from distributed computing).
Finally, inspired by the computation of the $\alpha$-core in Section~\ref{sec:min}, a very interesting question is to look at more involved problems w.r.t. emergent behavior from simple protocols.

%
%
%
 \bibliographystyle{splncs04}
 \bibliography{sirocco-network-dynamics}
%
%
%
%
%

\newpage
\clearpage
\pagenumbering{roman}
\renewcommand*{\thepage}{\roman{page}}
\appendix

\section{\texorpdfstring{\abd}{(a,b)-Dynamics}~with \texorpdfstring{$\alpha=\beta$}{a=b}~ and a Proper Potential Function on the Degrees} \label{app:degree}

In this case we study \abd where the potential of a pair of nodes is any symmetric non-decreasing function on the degrees of its two endpoints, as happens with Section~\ref{sec:min}. We prove stabilization as well as that the number of steps needed until stabilization is $O(n)$, assuming $\alpha=\beta$.
More formally, we define the potential of a pair $(u,v)$ to be $\mathcal{E}(u,v)=f(d_{G^{(t)}}(u),d_{G^{(t)}}(v))$, where $f$ is a $\textit{proper}$ (symmetric and non-decreasing in both variables) function. The scheduler $\mathcal{S}$ is fixed and contains all $\binom{n}{2}$ possible pairwise interactions. 

For the graph $G^{(t)}$, let $R^{(t)}(u,v)$ be an equivalence relation defined on the set of nodes $V$ for time $t$, such that $(u,v)\in R^{(t)}$ if and only if $d_{G^{(t)}}(u)=d_{G^{(t)}}(v)$. The equivalence class $R^{(t)}_i$ corresponds to all nodes with degree $d(R^{(t)}_i)$, where $i$ is the rank of the degree in decreasing order. This means that the equivalence class $R^{(t)}_1$ contains all nodes with maximum degree in $G^{(t)}$. Assuming that $n=|V|$, the maximum number of equivalence classes is $n-1$, since the degree can be in the range $[0,n-1]$ and no pair of nodes can exist that have degree $0$ and $n-1$ simultaneously. Let $|G^{(t)}|$ be the number of equivalence classes in graph $G^{(t)}$. 
Before moving to the proof, we give certain properties of the dynamic process that hold for all $t\geq 1$, that is they hold after at least one round of the process (initialization). These properties will be used in the proof for stabilization. 

From a bird eye's view of what follows, we notice that in this framework two nodes behave in the same way if their degrees are the same, due to the definition of the potential function. Furthermore, if at any time a node $u$ has degree at least as large as the degree of another node $v$, then it will form at least as many edges in the next time step, thus preserving the relative order of their degrees. These observations help us define some equivalence classes related to the degrees of the nodes, whose properties allow us to inductively prove our upper bounds. This intuition is formalized in the following properties:

\begin{prop}\label{prop:Monotonicity}
If $d_{G^{(t)}}(u) \geq d_{G^{(t)}}(w)$, then $d_{G^{(t+1)}}(u) \geq d_{G^{(t+1)}}(w)$, for all $t\geq 1$.
\end{prop}
\begin{proof}

For any neighbor $v$ of $w$ in $G^{(t+1)}$ it holds that $\mathcal{E}^{(t)}(v,w) \geq \beta$. Then it also holds that $\mathcal{E}^{(t)}(v,u) \geq \beta$, since $f$ is non-decreasing, which means $v$ is also a neighbor of $u$ in $G^{(t+1)}$. \end{proof}

\noindent Nodes that have the same degree at time $t$, share the same neighbors at time $t+1$. 

\begin{prop}\label{prop:EqualDegree}
If $d_{G^{(t)}}(u) = d_{G^{(t)}}(w)$, then $N_{G^{(t+1)}}(u) = N_{G^{(t+1)}}(w)$.
\end{prop}
\begin{proof}

As in the proof of Property~\ref{prop:Monotonicity}, due to the equality of the degrees, it also holds that any neighbor $v$ of $u$ is a neighbor of $w$ and respectively any neighbor $v$ of $w$ is a neighbor of $u$.
\end{proof}

\noindent In the following, we discuss properties related to equivalence classes.


\begin{prop} \label{prop:ECNonIncreasing}
The number of equivalence classes in $G^{(t+1)}$ is less than or equal to the number of equivalence classes in $G^{(t)}$.
\end{prop}
\begin{proof}

By Property~\ref{prop:EqualDegree}, nodes that belong to the same equivalence class at time $t>0$ will always belong to the same equivalence class for all $t'>t$.
\end{proof}


\begin{prop}  \label{prop:GraphsSameEC}
If $G^{(t+1)}$ has the same number of equivalence classes as $G^{(t)}$, then $\forall i$, $|R^{(t)}_i|=|R^{(t+1)}_i|$, where $|R^{(t)}_i|$ is the number of nodes in the equivalence class $R^{(t)}_i$.
\end{prop}
\begin{proof}

Suppose that the above does not hold. Then, there is some $i$ for which $|R^{(t)}_i| \neq |R^{(t+1)}_i|$. This means that there must be two nodes in some equivalence class $R^{(t)}_j$ that landed to different classes in $G^{(t+1)}$. However, Property~\ref{prop:EqualDegree} implies that this is impossible. 
\end{proof}

The following lemma shows how equivalence classes behave w.r.t. edge distribution.
\setcounter{lemma}{3}
\begin{lemma} \label{lem:MonotonicityEC}
If an arbitrary node $u$ in $R^{(t)}_i$ is connected with some node $w$ in $R^{(t)}_j$, then $u$ is connected with every node $x$ in every equivalence class $R^{(t)}_k$, such that $k\leq j$ and $t>0$.
\end{lemma}
\begin{proof}
Due to Property~\ref{prop:Monotonicity}, for all nodes $x\in R^{(t)}_k$ it holds that $d_{G^{(t)}}(x) \geq d_{G^{(t)}}(w)$ and so they are also neighbors of $u$.
\end{proof}

We prove by induction that this \abd always stabilizes in at most $|G^{(0)}|+1$ steps. To begin with, it is obvious that the clique $\mathcal{K}_n$ as well as the null graph $\overline{\mathcal{K}_n}$ both stabilize in at most one step, for any value of $\beta$. The following renormalization lemma describes how the number of equivalence classes is reduced and is crucial to the induction proof. 

\begin{lemma}	\label{lem:Renormalization}
If $d(R^{(t)}_1) = n-1$, $\forall t\geq c, c\in \mathbb{N}$, and the subgraph $G^{(c)} \setminus R^{(c)}_1$ stabilizes for any value of $\beta$ and proper function $f$, then $G^{(c)}$ stabilizes as well. Similarly, if $d(R^{(t)}_{|G^{(t)}|}) = 0$, $\forall t\geq c, c\in \mathbb{N}$, and the subgraph $G^{(c)} \setminus R^{(c)}_{|G^{(c)}|}$ stabilizes for any value of $\beta$ and proper function $f$, then $G^{(c)}$ stabilizes as well. The time it takes for $G^{(c)}$ to stabilize is the same as the time it takes for the induced subgraph to stabilize for both cases.
\end{lemma}
\begin{proof}
The main idea is that we consider two different sets of nodes: $R^{(c)}_1$ and $V\setminus R^{(c)}_1$. Due to our hypothesis, at all future time steps the edges between these two groups, and the edges with both endpoints in $R^{(c)}_1$ are fixed. Concerning the edges with both endpoints in $V\setminus R^{(c)}_1$, we can almost study this subgraph independently. That's because the effect of $R^{(c)}_1$ on $V\setminus R^{(c)}_1$ is completely predictable: it always increases the degree of all nodes by the exact same amount. The same reasoning applies for $R^{(c)}_{|G^{(c)}|}$.

More formally, by Property~\ref{prop:Monotonicity}, for all $t\geq c$ it holds that $R^{(t)}_1 \subseteq R^{(t+1)}_1$. This means that the nodes in $R^{(c)}_1$ are always connected to every node after time $c$. As a result, for all $u\in V\setminus R^{(c)}_1$ it holds that their degree in the induced subgraph $G^{(t)}\setminus R^{(c)}_1$ is $d_{G^{(t)}}(u)-|R^{(c)}_1|$. 
Thus, the decision for the existence of an edge $(u,v)$, where $u,v\in G^{(t)}\setminus R^{(c)}_1$ is the following:
\[\mathcal{E}^{(t)}(u,v)=f(d_{G^{(t)}\setminus R^{(c)}_1}(u)+|R^{(c)}_1|,d_{G^{(t)}\setminus R^{(c)}_1}(v)+|R^{(c)}_1|)\geq \beta\]
which can be written as: 
\[\mathcal{E}^{(t)}(u,v)=g(d_{G^{(t)}\setminus R^{(c)}_1}(u),d_{G^{(t)}\setminus R^{(c)}_1}(v))\geq \beta\]
where
\[g(x,y)=f(x+|R^{(c)}_1|,y+|R^{(c)}_1|)\]

Clearly, $g$ is a proper function assuming that $f$ is a proper function. Thus, the choice of whether the edge exists between $u$ and $v$ is equivalent between $G^{(t)}$ and $G^{(t)}\setminus R^{(c)}_1$ by  appropriately changing $f$ to $g$. But due to our hypothesis $G^{(c)}\setminus R^{(c)}_1$ stabilizes, and thus $G^{(c)}$ also stabilizes in the same number of steps. Note that we need not compute $g$ since this is only an analytical construction; the dynamic process continues as defined.
The proof of the second part of the lemma is similar in idea but much simpler since function $f$ does not change due to the fact that the removed nodes have degree $0$.
\end{proof}

The following theorem establishes that this \abd stabilizes in linear time. 
\setcounter{theorem}{4}
\begin{theorem} \label{thm:beta_degree}
When $\alpha=\beta$, $f$ is proper, $\mathcal{E}(u,v)=f(d_{G^{(t)}}(u),d_{G^{(t)}}(v))$, and the scheduler $C^{(t)}$ contains all $\binom{n}{2}$ possible pairwise interactions, \abd stabilizes on given $G^{(0)}$ in at most $|G^{(0)}|+1$ steps.
\end{theorem}
\begin{proof}
By Property~\ref{prop:ECNonIncreasing} we have that $|G^{(1)}|\le |G^{(0)}|$. Therefore, it suffices to prove that \abd stabilizes in at most $|G^{(1)}|+1$ steps, or equivalently that it stabilizes in at most $|G^{(1)}|$ steps after time $1$; for technical reasons, we prove that for any $t_0>0$, \abd stabilizes in at most $|G^{(t_0)}|$ steps after $t_0$. This is necessary for some of the needed tools to work (for example Lemma~\ref{lem:MonotonicityEC}, which doesn't work for time $0$).

We prove our claim inductively, on the number of equivalence classes at time $t_0$. For the base case, if $|G^{(t_0)}|=1$, then we have a regular graph. If $f(d(R^{(t_0)}_{1}), d(R^{(t_0)}_{1})) < \beta$, we get that $G^{(t_0+1)}$ is the null graph $\overline{\mathcal{K}_n}$, which indeed stabilizes because $f(d(R^{(t_0+1)}_{1}), d(R^{(t_0+1)}_{1})) = f(0,0) \le f(d(R^{(t_0)}_{1}), d(R^{(t_0)}_{1})) < \beta$. Similarly, if $f(d(R^{(t_0)}_{1}), d(R^{(t_0)}_{1})) \ge \beta$ we get that $G^{(t_0+1)}$ is the complete graph $\mathcal{K}_n$, which stabilizes because $f(d(R^{(t_0+1)}_{1}), d(R^{(t_0+1)}_{1})) = f(n-1,n-1) \ge f(d(R^{(t_0)}_{1}), d(R^{(t_0)}_{1})) \ge \beta$.

For the inductive step, suppose that $|G^{(t_0)}|>1$. If $|G^{(t_0+1)}| < |G^{(t_0)}|$, then the lemma follows by our inductive hypothesis. Else, we discern two cases, namely whether $f(n-1,0) < \beta$ or $f(n-1,0) \ge \beta$.

We begin with the case $f(n-1,0) < \beta$. If at some time step $t\ge t_0$ it holds that $d(R^{(t)}_{|G^{(t)}|}) = 0$, then for all $t'\ge t$ it still holds that $d(R^{(t')}_{|G^{(t')}|}) = 0$. To see this, notice that if it does not hold, then there exists a minimal $t'>t$ such that a node $u \in R^{(t)}_{|G^{(t)}|}$ has degree $d^{(t')}(u) > 0$. But this means that there exists some vertex $v \neq u$ such that $f(d^{(t'-1)}(v), d^{(t'-1)}(u)) = f(d^{(t'-1)}(v),0) \ge \beta$. But since $d^{(t'-1)}(v) \le n-1$, and $f(n-1,0) < \beta$, we reach a contradiction.

By the above observation and Lemma~\ref{lem:Renormalization}, it immediately follows that if $d(R^{(t_0)}_{|G^{(t_0)}|}) = 0$ or $d(R^{(t_0+1)}_{|G^{(t_0+1)}|}) = 0$, then our lemma holds.

Therefore, we are only left with the case where $|G^{(t_0+1)}| = |G^{(t_0)}|$ and no node has degree $0$, neither in $G^{(t_0)}$ nor in $G^{(t_0+1)}$. For any $i$, the $i$-th equivalence class of $G^{(t_0)}$ and the $i$-th equivalence class of $G^{(t_0+1)}$ have the same number of nodes, by Property~\ref{prop:GraphsSameEC}. If they also have the same degree, then Lemma~\ref{lem:MonotonicityEC} shows that the two graphs are equal, and thus we have stabilization in $0$ steps.

By Lemma~\ref{lem:MonotonicityEC}, each of the $|G^{(t_0)}|$ equivalence classes at time $t_0$ has only $|G^{(t_0)}|+1$ possible values for its degree, and, by definition, no two classes have the same degree. However, one of these values is $0$, which we ruled out for any equivalence class, meaning that there are only $|G^{(t_0)}|$ possible values for the $|G^{(t_0)}|$ pairwise disjoint degrees. The same argument can be made for $t_0+1$. However, by Property~\ref{prop:GraphsSameEC}, we get that the possible values for both time steps are the same, concluding that for all $i \in \{1, \ldots, |G^{(t_0)}|\}$, we have $d(R^{(t_0)}_i) = d(R^{(t_0+1)}_i)$.

The case $f(n-1,0) \ge \beta$ is completely similar. If at some time step $t\ge t_0$ it holds that $d(R^{(t)}_{1}) = n-1$, then for all $t'\ge t$ it still holds that $d(R^{(t')}_{1}) = n-1$. To see this, notice that if it does not hold, then there exists a minimal $t'>t$ such that a node $u \in R^{(t)}_{1}$ has degree $d^{(t')}(u) < n-1$. But this means that there exists some vertex $v \neq u$ such that $f(d^{(t'-1)}(u), d^{(t'-1)}(v)) = f(n-1,d^{(t'-1)}(v)) < \beta$. But since $d^{(t'-1)}(v) \ge 0$, and $f(n-1,0) \ge \beta$, we reach a contradiction.

By the above observation and Lemma~\ref{lem:Renormalization}, it immediately follows that if $d(R^{(t_0)}_{1}) = n-1$ or $d(R^{(t_0+1)}_{1}) = n-1$, then our lemma holds. Therefore, we are only left with the case where $|G^{(t_0+1)}| = |G^{(t_0)}|$ and no node has degree $n-1$, neither in $G^{(t_0)}$ nor in $G^{(t_0+1)}$. 

Therefore, we are only left with the case where $|G^{(t_0+1)}| = |G^{(t_0)}|$ and no node has degree $0$, neither in $G^{(t_0)}$ nor in $G^{(t_0+1)}$. For any $i$, the $i$-th equivalence class of $G^{(t_0)}$ and the $i$-th equivalence class of $G^{(t_0+1)}$ have the same number of nodes, by Property~\ref{prop:GraphsSameEC}. If they also have the same degree, then Lemma~\ref{lem:MonotonicityEC} shows that the two graphs are equal, and thus we have stabilization in $0$ steps.

By Lemma~\ref{lem:MonotonicityEC}, each of the $|G^{(t_0)}|$ equivalence classes at time $t_0$ has only $|G^{(t_0)}|+1$ possible values for its degree, and, by definition, no two classes have the same degree. However, one of these values is $n-1$, which we ruled out for any equivalence class, meaning that there are only $|G^{(t_0)}|$ possible values for the $|G^{(t_0)}|$ pairwise disjoint degrees. The same argument can be made for $t_0+1$. However, by Property~\ref{prop:GraphsSameEC}, we get that the possible values for both time steps are the same, concluding that for all $i \in \{1, \ldots, |G^{(t_0)}|\}$, we have $d(R^{(t_0)}_i) = d(R^{(t_0+1)}_i)$.

\end{proof}

\section{Turing-Completeness} \label{app:sec_tc}
\subsection{Cellular Automata and Rule 110}  \label{app:ssec_ca-rule110}
A one-dimensional cellular automaton, or, as called by Wolfram, an elementary cellular automaton, is a discrete model of computation. It consists of a one-dimensional grid of infinitely many cells, each containing a binary value. The value of all cells is updated synchronously, in discrete time steps. Each cell updates its value based on its own value and the values of its two neighboring cells.

Since the new value of each cell depends on $3$ binary values, there are only $8$ different cases for this update. We write $001$ for the case where the left neighbor's value and the current value of a cell is $0$ while the right neighbor's value is $1$, $101$ for the case where both neighbors have value $1$ while the current value is $0$, and so on. Wolfram proposed the following numbering scheme for elementary cellular automata. Suppose we create a binary number whose most significant bit is the updated value of a cell in case $111$, the second most significant bit is the updated value in case $110$, and so on until the least significant bit, the updated value in case $000$. If we acquire number $X$ by translating this binary number to decimal, then this particular cellular automaton is \emph{Rule $X$}.

Therefore, Rule $110$ is the cellular automaton corresponding to the binary number $01101110$; simply put, the updated value of a cell is equal to its right neighbor's value, if its current value is $0$. Else, it is $0$ iff both its neighbors have value $1$. What is interesting about Rule $110$ is that although it is very easy to describe, Cook proved it to be Turing-Complete \cite{DBLP:journals/compsys/000104}. One shall think of the initial configuration of the cells to contain both the program and its input; if the Turing machine corresponding to the program would halt on this input, then Rule $110$ stabilizes to a state that keeps on repeating forever. From this state, one is able to directly retrieve what the Turing machine would output. This allows us to prove Turing Completeness for some model of computation by just showing that it is able to simulate Rule $110$, which is much simpler than a Turing machine.

\subsection{Proofs of Turing Completeness section}  \label{app:ssec_tc-proofs}
For reference in the proofs that follow, Figure~\ref{fig:CG} depicts how $CG(i)$ (cell $i$) is connected to $CG(i+1)$ (cell $i+1$) and $CG(i-1)$ (cell $i-1$). 

\begin{figure}[t]
\centering
\includegraphics[scale=0.14]{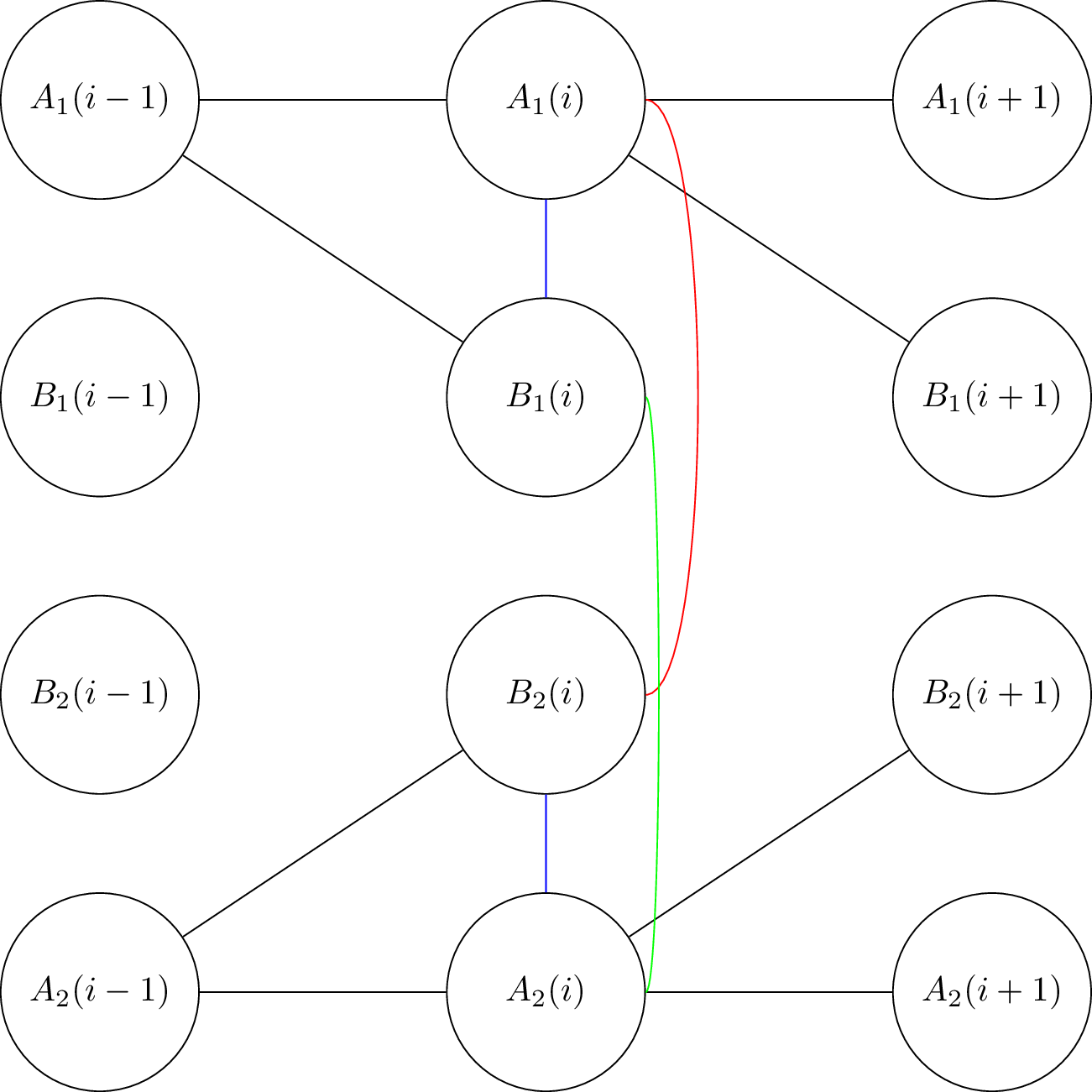}
\caption{Each circle represents a $PCG$ and each line represents a connection between $PCG$s ($4$ always-on gadgets) as in Figure~\ref{fig:gadgets}. Only connections relevant to $A_1(i),A_2(i),B_1(i),B_2(i)$ are shown. The $4$ connections in the second column (again each one is $4$ always-on gadgets) are internal connections of $CG(i)$. All other connections correspond to how $CG(i-1)$ is connected with $CG(i)$ and $CG(i)$ is connected with $CG(i+1)$. We prove that these connections are always preserved.}
\label{fig:CG}
\end{figure}

\setcounter{lemma}{9}
\begin{lemma} \label{app:tc-main-structure}
If there exists a flip $(x,y)$-gadget connected to an $A_j(i)$ $PCG$ in $G^{(0)}$, then the edge $(x,y)$ at time $t$ exists if and only if $t \in \mathbb{N}\cup \{0\}$. Similarly, if there exists a flip $(x,y)$-gadget connected to a $B_j(i)$ $PCG$ in $G^{(0)}$, then the edge $(x,y)$ exists if and only if $t \not \in \mathbb{N}\cup \{0\}$. Finally, all other edges exist at any time step if and only if they exist in $G^{(0)}$, with the exception of edges between $(h,l)$ nodes of a $PCG$.
\end{lemma}
\begin{proof}
We prove our claim using induction on the time step $t$. The base case $t=0$ holds by the construction of $G^{(0)}$. Suppose our claim holds for time step $t-0.5$, we show that it also holds for time step $t$. We first prove our claim for the pairs of nodes sharing an edge in $G^{(0)}$, except for the pairs $(h,l)$ of $PCG$s, as the Lemma makes no claim about them. Notice that it suffices to argue about always-on and flip gadgets, as this is the only way we added non-$(h,l)$ edges to $G^{(0)}$.

Let us first focus on the nodes that, at $G^{(0)}$, are contained in the same always-on $(x,y)$-gadget. We argue that for any two such nodes $x',y'$, the edge between them exists on time step $t$, except possibly for the $(x,y)$ edge; more formally, the unordered pair $\{x',y'\}$ is assumed to be different from $\{x,y\}$. By definition of the always-on gadget and the inductive hypothesis, $x'$ and $y'$ have exactly $20$ common neighbors in $G^{(t-0.5)}$, and thus they continue sharing an edge in $G^{(t)}$. Concerning the $x,y$ nodes of the gadget, we take cases depending on whether they also happen to be the two special endpoints of a flip $(x,y)$ gadget in $G^{(0)}$ or not. In the former case, by the inductive hypothesis, they have between $40$ and $41$ common neighbors in $G^{(t-0.5)}$, depending on the existence of edges not defined by our induction hypothesis. Thus, these edges always flip their status at $t$, as the lemma dictates. In the latter case they have between $20$ and $24$ common neighbors in $G^{(t-0.5)}$, depending on the existence of edges not defined by our induction hypothesis. Thus, these edges continue to exist in $G^{(t)}$. 

We are only left to argue about pairs of nodes with no edge connecting them in $G^{(0)}$. For a non-existent edge to become existent, it must be that its two endpoints have at least $40$ common neighbors, by the potential function. But, by the inductive hypothesis and the construction of $G^{(0)}$, this only happens for endpoints $x,y$ for which there exists a flip $(x,y)$-gadget (we already argued about such cases) and for endpoints $h,l$ of some $PCG$ (for which case our lemma does not claim anything). Thus, no other edge is ever created.
\end{proof}

\setcounter{lemma}{10}
\begin{lemma} \label{app:Simulate_110}
It holds that $A_j^{(t)}(i)=B_j^{(t)}(i)=cell^{(t)}(i)$ for $j\in\{1,2\}$ and all $i,t\in \mathbb{N}$.
\end{lemma}
\begin{proof}
It holds that $A_j^{(0)}(i)=B_j^{(0)}(i)=cell^{(0)}(i)$ by the initialization of our construction. Suppose that $A_j^{(t)}(i)=B_j^{(t)}(i)=cell^{(t)}(i)$ for an integer $t\geq 0$. By using induction we show that the lemma holds for time $t+1$.

First of all, we prove that $A_j^{(t+0.5)}(i)=cell^{(t+1)}(i)$. If $cell^{(t)}(i)=0$, then it holds that $cell^{(t+1)}(i)=cell^{(t)}(i+1)=A_j^{(t)}(i+1)=B_j^{(t)}(i+1)$, due to our inductive hypothesis. Furthermore, due to our inductive hypothesis it holds that $A_j^{(t)}(i)=B_1^{(t)}(i)=B_2^{(t)}(i)=0$. Thus, since $CN^{(t)}(A_j(i)(h),A_j(i)(l))=70$ and $|E^{(t)}(A_j(i)(h),A_j(i)(l))|=0$ (there is no edge between the $(h,l)$ nodes in $A_j(i)$) the potential between the pair of nodes is $\mathcal{E}^{(t)}(A_j(i)(h),A_j(i)(l))=CE^{(t)}(A_j(i)(h),A_j(i)(l))+\beta-10$. 
To find the potential of the pair of nodes $A_j(i)$ we compute:
\[CE^{(t)}(A_j(i)(h),A_j(i)(l))=8+A_j^{(t)}(i-1)+B_1^{(t)}(i)+B_2^{(t)}(i)+A_j^{(t)}(i+1)+B_j^{(t)}(i+1)=\]
\[8+cell^{(t)}(i-1)+2cell^{(t)}(i+1)\]
Thus, it follows that the potential of $A_j(i)(h)$ and $A_j(i)(l)$ is $\beta+cell^{(t)}(i-1)+2cell^{(t)}(i+1)-2$, which is at least $\beta$ if and only if $cell^{(t)}(i+1)=1$. Thus, in the case where $cell^{(t)}(i)=0$ we proved that indeed it holds that $A_j^{(t+0.5)}(i)=cell^{(t+1)}(i)$.

We use a similar reasoning for the case where $cell^{(t)}(i)=1$. In particular, since $CN^{(t)}(A_j(i))=70$ and $|E^{(t)}(A_j(i))|=1$ (there is an edge between the $(h,l)$ nodes in $A_j(i)$) the potential between the pair of nodes is $\mathcal{E}^{(t)}(A_j(i)(h),A_j(i)(l))=\beta +12-CE^{(t)}(A_j(i))$. We compute:
\[CE^{(t)}(A_j(i)(h),A_j(i)(h))=8+A_j^{(t)}(i-1)+B_1^{(t)}(i)+B_2^{(t)}(i)+A_j^{(t)}(i+1)+B_j^{(t)}(i+1)=\]
\[=10+cell^{(t)}(i-1)+2cell^{(t)}(i+1)\]
Thus, it follows that the potential of $A_j(i)(h)$ and $A_j(i)(l)$ is $\mathcal{E}^{(t)}(A_j(i))=\beta+2-cell^{(t)}(i-1)-2cell^{(t)}(i+1)$, which is less than $\beta$ if and only if $cell^{(t)}(i-1)=cell^{(t)}(i+1)=1$. This proves that $A_j^{(t+0.5)}(i)=cell^{(t+1)}(i)$.

It also holds that $A_j^{(t+1)}(i)=cell^{(t+1)}(i)$, because $CN^{(t+0.5)}(A_j(i)(h),A_j(i)(l))=10$, and thus $A_j^{(t+1)}(i)=A_j^{(t+0.5)}(i)$. Similarly, $B_j^{(t+0.5)}(i) = B_j^{(t)}(i)$ as $CN^{(t)}(B_j(i)(h),B_j(i)(l))=6$.

The potential of $B_j(i)$ at time $t+0.5$ is (recall that $CN^{(t)}(B_j(i)(h),B_j(i)(l))=66$):
\[\mathcal{E}^{(t+0.5)}(B_j(i)(h),B_j(i)(l))=CE^{(t+0.5)}(B_j(i)(h),B_j(i)(l))+\beta-6=\] \[\beta+2A_j^{(t+0.5)}(i)+A_j^{(t+0.5)}(i-1)-2\]
This is at least $\beta$ if and only if $A_j^{(t+0.5)}(i)=1$, which proves that $B_j^{(t+1)}(i)=cell^{(t+1)}(i)$. 
\end{proof}

\setcounter{theorem}{12}
\begin{theorem} \label{app:tc}
The \abd is Turing-Complete.
\end{theorem}
\begin{proof}
By Lemma~\ref{lem:tc-main-structure} and Corollary~\ref{cor:eq} it follows that Rule $110$ would be correctly simulated by the particular \abd constructed above, if the transitional non-integer time steps were missing, and thus the convergence of an instance of Rule $110$ would mean the stabilization of the constructed \abd. To achieve this, we simulate the two steps of the constructed \abd in one step based on the observation that the defined potential for each pair of nodes $x,y$ depends only on the graph induced by the nodes at distance at most $1$ from either $x$ or $y$.
As a result, if nodes $x$ and $y$ at time step $t$ could 'guess' what this induced graph would look like in the transitional, non-integer, time step $t+0.5$, they could immediately use this to deduce their potential in time step $t+0.5$. 

We are left to argue about how $x$ and $y$ get information about this induced graph. Notice that a node $u$ may get connected with another node $v$ at any time step $t'$ only if $d^{(t'-0.5)}(u,v) \le 2$. Thus, in order for $x$ and $y$ to be able at time step $t$, to know this induced graph at time step $t+0.5$, it suffices to compute the connections at time $t+0.5$ between all nodes $u$ for which $\min\{d^{(t)}(x,v), d^{(t)}(y,v)\} \le 2$. In turn, in order to compute such a potential, they need to have information about nodes at distance $1$ from these nodes that lie at distance at most $2$. In conclusion, it suffices to access all nodes at distance at most $3$ at time $t$; notice that by Lemma~\ref{lem:tc-main-structure} and the construction of $G^{(0)}$, there is a constant number of such nodes, for any pair $(x,y)$ and time $t$.

Therefore, the new \abd starts with the same $G^{(0)}$ and computes the new potential between any two nodes $x,y$ in two conceptual steps. In the first step, it uses the old potential function, and information from nodes at distance at most $3$ from either of them, to compute how the graph induced by all nodes $u$ for which $\min\{d^{(t)}(x,u), d^{(t)}(y,u)\} \le 2$ would look like at time $t+0.5$. Then, by applying the old potential function on this computed graph, it computes the final potential between $x$ and $y$, effectively simulating the transitional time step.  Therefore, the potential function only acquires information from nodes at a constant distance (at most $3$) from either $x$ or $y$, as required. It is also clear that it is network-agnostic, or in other words that it is designed without access to the topology of $G^{(0)}$.

To see that this new potential function is computationally symmetric, notice that the auxiliary graph is computed both by $x$ and by $y$ by accessing the same information and using the same computationally symmetric potential function, meaning both $x$ and $y$ end up with the same auxiliary graph. Then, they apply the same computationally symmetric function on this graph, meaning that they acquire the same value.

Finally, we have shown that at any time step, each node only has a constant number of neighbors. Therefore, the auxiliary graph also has a constant number of nodes, and we only need a constant number of words to represent the auxiliary graph. The computation of each such edge in the auxiliary graph, as well as the final computation, uses the old potential function; all these computations are using the same working memory. Thus, the new potential function respects the restriction of having a working memory at most (asymptotically) logarithmic in size, compared to the input memory (which contains the neighborhoods of $u$ and $v$), since the old potential function does as well. The time needed is also polynomial in the input size, as the same holds for the time needed to compute the old potential function.

\end{proof}

\end{document}